\documentclass[11pt,a4paper]{article}
\usepackage{latexsym}
\usepackage{amsfonts}
\usepackage[mathscr]{eucal}
\usepackage{epsfig}
\usepackage{a4wide}
\usepackage{ifthen}
\usepackage{comment}

\usepackage{amsthm,amsmath,natbib}

\usepackage{amsmath, amsthm}

\def\startlocaldefs{\makeatletter}
\startlocaldefs

\def\bysame{\@ifnextchar\bgroup{\@bysame}{\@bysame{}}}
\def\@bysame#1{\vrule height 1.5pt depth -1pt width 3em \hskip
0.5em\relax}

\def\@copyright{}
\newcommand{\UCL}{\ensuremath{\operatorname{UCL}}}
\newcommand{\LCL}{\ensuremath{\operatorname{LCL}}}

\newcommand{\ve}{\varepsilon}
\newcommand{\eins}{\mathbf{1}}
\newcommand{\N}{\mathbb{N}}
\newcommand{\R}{\mathbb{R}}
\newcommand{\calF}{\mathcal{F}}
\newcommand{\calI}{\mathcal{I}}
\newcommand{\calJ}{\mathcal{J}}
\newcommand{\calN}{\mathcal{N}}
\newcommand{\calT}{\mathcal{T}}
\newcommand{\calV}{\mathcal{V}}
\newcommand{\Cov}{\operatorname{Cov}}

\newcommand{\E}{\mathbb{E}}
\newcommand{\Var}{\operatorname{Var\,}}
\newcommand{\PP}{\mathbb{P}}

\newcommand{\trunc}[1]{\lfloor #1 \rfloor}

\newtheorem{theorem}{Theorem}[section]

\newtheorem{example}{Example}[section]

\newtheorem{remark}{Remark}[section]

\excludecomment{details}

\begin{document}



\begin{center}
  \Large
  A BINARY CONTROL CHART TO DETECT SMALL JUMPS
\end{center}
\vskip 1cm
\begin{center}
Ewaryst Rafaj\l owicz\\
Institute of Computer Engineering, Control and Robotics\\
Wroc\l aw University of Technology, Poland
\end{center}
\begin{center}
  Ansgar Steland\footnote{Address of correspondence: Prof. Dr. A. Steland, RWTH Aachen University, Institute of Statistics, W\"ullnerstr. 3,
  D-52056 Aachen, Germany.}\\
  Institute of Statistics\\
  RWTH Aachen University, Germany\\ \ \\
  July 16th, 2008

\end{center}

\begin{abstract}
The classic $N \,p$ chart gives a signal if the number of
successes in a sequence of independent binary variables exceeds a
control limit. Motivated by engineering applications in industrial
image processing and, to some extent, financial statistics, we
study a simple modification of this chart, which uses only the
most recent observations. Our aim is to construct a control chart
for detecting a shift of an unknown size, allowing for an unknown
distribution of the error terms. Simulation studies indicate that
the proposed chart is superior in terms of out-of-control average
run length, when one is interest in the detection of very small
shifts. We provide a (functional) central limit theorem under a
change-point model with local alternatives which explains that
unexpected and interesting behavior. Since real observations are
often not independent, the question arises whether these results
still hold true for the dependent case. Indeed, our asymptotic
results work under the fairly general condition that the
observations form a martingale difference array. This enlarges the
applicability of our results considerably, firstly, to a large
class time series models, and, secondly, to
locally dependent image data, as we demonstrate by an example.  
\end{abstract}


\section{Introduction} Detection of changes in the mean
characteristic of produced items is still the most frequently used
tool in quality control. A large variety of control charts have
been proposed in the last fifty years. For comprehensive reviews
we refer to \cite{AnHu02a}, \cite{AnHu02b}, the monograph
\cite{BD}, and also to the articles \cite{AttribChRev},
\cite{nonparchartrev}, and \cite{Montgomery}. Investigations of
their properties indicate that one can not hope to select one
"universally good" chart, which is uniformly sensitive to small,
moderate and large shifts in the mean and still robust against
violating the normality of errors assumption. On the other hand, a
wide accessability of computer systems allows to run
simultaneously several control charts with different sensitivity
ranges for the same process. It is well known, the Shewart chart
is well tuned to detect rather quickly large shifts, while EWMA
and CUSUM charts are faster in detecting smaller shifts of the
order $0.5\sigma$. If the aim is to detect moderate to large jumps
so called jump-preserving procedures are attractive, which are
special cases of the unifying vertically weighted regression
approach studied by \cite{VWRRaf99}, \cite{Ste05ClipMed}, and
\cite{VWR_PRS04, VWR_PRS08}. Nonparametric kernel control charts
and the optimization for certain out-of-control models covering
mixing processes have been studied in \cite{Ste04Metrika} and
\cite{Ste05KernelSmoothers}. Further, \cite{smallshift} combined a
classic Shewhart chart and a conforming run length chart yielding
smaller ARLs for shifts larger than $ 0.8 \sigma $, but that
method is inferior to the EWMA chart for smaller shifts.

The purpose of this paper is to propose a new binary chart, which
is easy to apply, has enlarged sensitivity to very small shifts,
and is robust with respect to deviations from normality. We
provide a comprehensive study covering the methodology, asymptotic
theory, practical issues of control chart design, and extensive
Monte Carlo simulations.

Our study is motivated as follows: Although computing power has
considerably increased, many practical applications still require
detection procedures which are extremely fast to calculate. An
example, which motivated our investigation, is the surveillance of
copper production as outlined in \cite{VWR_PRS08}. Here the
problem is to detect defects and cracks resulting in lower
quality. The copper is surveyed by a camera taking many
high-resolution images per second, and each column of an image is
analyzed in real time to detect defects. Only detectors which are
fast enough to calculate can be employed. In such engineering
image processing and image analysis applications one has to deal
with the spatial inhomogeneity of the grey level of pixels. One
can either assume that the inhomogeneity is compensated by a quite
wiggly mean function which is disturbed by independent noise, or
assume a smooth mean function overlayed by dependent noise. In the
latter case fitting complex models to take account of dependencies
is often not feasible in real-time applications. Then it is
important to know how the chosen method behaves for dependent
data. Let us also mention a further important area, namely the
application of monitoring procedures to financial data. In
financial statistics various empirical analyzes have revealed that
asset returns are usually uncorrelated but the squares are
serially correlated and are affected by conditional
heteroscedasticity which produces the clusters of strongly
dispersed returns seen in real data. Various models for returns
assume or imply the martingale difference property.

Having in mind the above applications, we propose a simple method
where one thresholds the observations to obtain binary data and
applies a control chart based on the number of data points exceeding
the threshold. In contrast to the classic $N \, p$-chart, the chart
uses a finite {\em buffer} storing only the most recent
observations. Our simulation results indicate that such a modified
$p$-chart with a reduced number of observations reacts on average
slower than several control charts studied recently in
\cite{AnnStat2004} for shifts larger than $ 0.25 \sigma $, but
provides faster detection for very small shifts.

We provide an appropriate theoretical framework and prove a
functional central limit theorem which shows that the classic $ N \,
p $ chart's sensitivity with respect to very small shifts indeed can
be improved by taking less observations into account. As argued
above, the question arises, whether the result still holds true when
the independence assumption underlying the classic $p$ chart is
dropped. The answer is positive: Our main result and its
interpretation holds true for a large class of dependent processes,
namely the class of triangular arrays of random variables forming a
martingale difference array with respect to some filtration. Thus,
the benefits of the modified $p$ chart are also effective when
monitoring dependent data.

The paper is organized as follows. In Section~\ref{TheChart}, we
introduce the proposed control chart and its relationship to the
classic $N \, p$-chart. An appropriate change-point model with
local alternatives is introduced in Section~\ref{CPM} to study the
problem from an asymptotic viewpoint. We establish a functional
central limit theorem for the underlying stochastic process which
induces the stopping time of interest. A proof of the main result
is postponed to an appendix. Practical issues of control chart
design are discussed in detail in Section~\ref{PracticalIssues}.
Finally, an extensive Monte Carlo study is presented in
Section~\ref{Sim} providing a comparison with recently proposed
control charts.

\section{Statistical model and a modified $p$-chart}
\label{TheChart}

Our aim is to construct a control chart for detecting a shift of
an unknown size $m$ allowing for an unknown distribution $F$ of
the error terms. It is required that the in-control average run
length (in-control ARL) of the chart can be tuned to sufficiently
large values in order to reduce the number of false alarms.
Simultaneously, the out-of-control ARL should be small, leading to
quick detection of the jump after its occurrence. For a discussion
of the design of control limits and their relationship to alarm
rates and ARLs we refer to \cite{AlarmRates}.

Even if the underlying distribution is normal, the Shewhart
control chart is not powerful for detecting small changes, say $m$
of the order of $ 0.1 \sigma $ to $ 0.25 \sigma $, if $ \sigma $
denotes the standard deviation of the errors. The EWMA
(exponentially weighted moving average) control chart is better
suited to this purpose, but its performance is still not
satisfactory in the range of very small shifts. For this reason a
number of modifications of the Shewhart, EWMA, and CUSUM charts
have been proposed recently (see \cite{AnnStat2004} and the
bibliography cited therein). However, the design of a concrete
control procedure with specific properties requires knowledge of
the error distribution.

\subsection{Change-point model}

In this paper, we consider a classic change-point model, where the
observations are of the form
\begin{equation}\label{obs}
    Y_n\: = Y+\: m\cdot \mathbf{1}(n-q) + \ve_n,\quad
    n\: =\: 1,2, \ldots
\end{equation}
$Y$ denotes the desired level of quality (target value) which is
disturbed by random errors $\ve_n$'s. The deterioration of quality
is modelled by jump (permanent shift in the quality
characteristic) of height $m \neq 0$, which appears at time
instant $q>0$. $q$ is called {\em change-point} and is assumed to
be non-stochastic but unknown. $\mathbf{1}(t)$ denotes the
indicator function on the set $ [0,\infty) $, i.e.,
\begin{equation}\label{jump}
    \mathbf{1}(t)\: =\: \left\{
    \begin{array}{ccc}
      0 & if & t<0\\
       1 & if & t\geq 0    \\
    \end{array} \right. .
\end{equation}
Thus, starting at the change-point $q$ there is a jump of height $ m
$. In Section~\ref{CPM} we consider a change-point model allowing
for jump sizes tending to $0$ at a certain rate.

To simplify the exposition, we shall assume $ Y = 0 $ in what
follows. For the same reason, let us tentatively assume that the
error terms  $\ve_n$ in \eqref{obs} are independent and identically
distributed random variables. That assumption will be relaxed in the
next section. Whereas classic procedures are restricted to normally
distributed noise, we allow for arbitrary distribution functions $F$
which are symmetric about $0$, i.e.,
\begin{equation}\label{sym}
  F(x)\: =\: 1-F(-x),\quad x\in \R.
\end{equation}
Particularly, we allow for distributions having no finite
expectations, e.g., the Cauchy distribution which has heavier
tails than the normal distribution, or the Laplace (double
exponential) law with lighter tails. Note that we do not require
the error terms to possess a density $f$, but if they do,
\eqref{sym} implies $ f(x)=f(-x) $.

\subsection{The binary control chart revisited}

The classic nonparametric $N \, p$-chart is distribution-free under
quite general assumptions, and therefore is applicable when the
error distribution is unknown. Although we confine our discussion to
the case that the change from the in-control to the out-of-control
scenario is given by a sharp jump, our approach can also be used for
more general scenarios, because the construction of the control
chart does not require knowledge of the underlying error
distribution. As we shall see below, the chart proposed in this
article provides noticeably smaller out-of-control ARL than the
classical and recently proposed control charts, but only for very
small shifts, which are of the order 0.1-0.25 standard deviation --
or its equivalent, based on the interquartile range, if the variance
does not exists. A large number of theoretical investigations and
computer simulations are witness of the fact that
  one can not expect existence of one "universal" chart with
  best performance in the whole range of shifts in the mean,
  if underlying distribution jump height are not specified.
From this point of view, the binary chart occupies the region of
small shifts.

Let us briefly review the definition and basic properties of the
classic $N \, p$-chart. Obviously, if the process \eqref{obs} is
in-control and \eqref{sym} holds, then -- roughly -- half of the
observations should be positive and the rest are expected to be
negative. In other words, having $N>1$ observations
\begin{equation}\label{sign}
Z_n\:\stackrel{def}{=}\: \mbox{sign}(Y_n)\:=\: \left\{
    \begin{array}{ccc}
      0 & if & Y_n < 0\\
      1 & if & Y_n\geq 0    \\
    \end{array} \right. ,\quad n=1,2,\ldots ,N
\end{equation}
and introducing the counting random variable
\begin{equation}\label{numbposit}
I_N\: \stackrel{def}{=}\: \mbox{card}\{Z_i=1,\: i=1,2,\ldots
,N\}=\sum_{i=1}^{N} Z_i
\end{equation}
we have $\E(I_N)=N/2$, since $I_N$ is a binomial random variable
corresponding to $N$ trials and success probability $ p_0 = 1/2 $.
Here and in the sequel $\E$ denotes the expectation.

If a shift of size $m$ occurred, then the distribution of
subsequent $Y_n$'s is no longer symmetric around zero and the
probability of $Z_n=1$ changes to
\begin{equation}\label{prob1}
    p_1\:= 1-F(-m)
\end{equation}
where $p_1$ can be larger or smaller than $1/2$, depending on
whether $m$ is positive or negative. Summarizing, one can detect a
shift $m$ by testing the  hypothesis $H_0: p_0=1/2$ against the
alternatives that the success probability in one trial is
different than $1/2$.

If the process is in-control, the dispersion of the binomial r.v.
$I_N$ equals $\sqrt{N\,p_0\,(1-p_0)}$. Then, $I_N/N$ has expectation
$ p_0 $ and dispersion $\sqrt{p_0\,(1-p_0)/N}$. Approximating the
binomial distribution by the corresponding normal law we arrive at
the well known $N \, p$-chart with upper control limit
\begin{equation}\label{ucl}
   \UCL = p_0\: +\: k\, \sqrt{p_0\,(1-p_0)/N}
\end{equation}
and the lower control limit ($\LCL$)
\begin{equation}\label{lcl}
   \LCL = p_0\: -\: k\, \sqrt{p_0\,(1-p_0)/N},
\end{equation}
where $k$ is selected according to required averaged run length
(ARL) in-control, the standard choice being $k=3$. If $I_N/N$ is
outside the interval $(\LCL,\UCL)$, then the out-of-control state
is claimed. Repeating the above reasoning, we can obtain the $N\,
p_0$ version of this chart with the following control limits for
$I_N$
\begin{equation}\label{UCLLCL}
    N p_0 \: \pm \: k\, \sqrt{N\, p_0\,(1-p_0)},
\end{equation}
where $k$ is selected as above. For further discussions we refer
to \cite{Montgomery}.

\subsection{Modified $p$ chart}

The above chart is the starting point for our modifications. They
are necessary, since the classical chart \cite[][pp.
284-294]{Montgomery} is based on counting nonconforming items in
samples of size $N$, which are either taken daily or at $ N $
consecutive days, if only one observation is available at each
day. In the latter case, which is the setting we have in mind, the
chart is applied only each $N$th time instance. This can yield
substantially larger delays in detection. Obviously, such sampling
schemes are not appropriate for our purposes. Thus, we shall
modify the chart in such a way that it counts a fixed number,
$M>1$ say, previous individual observations $Z_n=1$ in a moving
window. If the process is in-control, then we expect that about
$M/2$ observations correspond to $ Z_n = 1 $.

More formally, we form a finite buffer of the length $M$, which
contains only $M$ past observations, excluding the latest one
$Z_n$. $M$ is called {\em buffer length}. When observation $Z_{n}$
is available, it replaces $Z_{n-1}$, which is pushed to replace $
Z_{n-2} $ and so on. At each time instant $n$ the present buffer
contents is used to verify whether the process is in-control. To
fix this idea, define the number of positive observations
contained in the buffer in time $n$
\begin{equation}\label{bufcont}
    J_n\: = \: \mbox{card}\{Z_i=1, i=(n-1), \ldots, n-M\}
    =\sum_{i=n-M}^{n-1} Z_i.
\end{equation}
Note that the difficulty with an initial content of the buffer
appears. The proposed modified $p$-chart is built on the
assumption that historical pre-run data are available which are
known to form a random sample of the in-control process. Thus, in
the sequel we assume that at time $n=0$ the buffer contains past
observations of the in-control process, which are numbered as
$Z_{-1}, \ldots, Z_{-M}$. Formally, we start the chart at $n=0$,
when the observation $Z_0$ arrives. Then, for $n=1,2,\ldots $ it
is verified whether the control statistic $J_n$ lies between the
control limits
\begin{equation}\label{modUCL}
   \UCL\:=\: M\, p_0\: + \: k\, \sqrt{M\, p_0\,(1-p_0)},
\end{equation}
and
\begin{equation}\label{modLCL}
   \LCL\:=\: M\, p_0\: - \: k\, \sqrt{M \, p_0\,(1-p_0)}.
\end{equation}
Clearly, for $ p_0 = 1/2 $ these formulas simplify to $ \UCL =
M/2\:+\: k\,\sqrt{M}/2 $ and $ \LCL = \: M/2\: -\:k\, \sqrt{M}/2
$. If $J_n$ is smaller than \LCL or larger than \UCL, then
out-of-control state is signaled. Note that the difference between
\UCL and \LCL is constant for this chart.

The main difference between the proposed chart and the classical
one can be summarized as follows. The classical $N\,p$ chart is
based on samples of size $N$ from {\em non-overlapping} production
intervals. In contrast, our chart counts events $Z_n=1$ in the
buffer on length $M$, which is moving forward with $n$, in such a
way that new observation $Z_n$ enters the buffer, while the oldest
one is pushed out of it. In other words, the content of the buffer
at time $n$ and at time $n+1$ highly overlap.

\section{Asymptotic results}
\label{CPM}

We will now present some asymptotic theory for the proposed
procedure providing an explanation of the superiority of the
modified $p$ chart for small jumps. To simplify exposition, we
slightly change the setting: We confine our study to a truncated
version of the one-sided control chart which gives a signal if $
J_n $ exceeds $ \UCL $ for some $ 1 \le n \le N $. However, our
results can be extended to deal with the general case as outlined
in \cite{Steland08}. The small jump setting will be modelled by an
appropriate asymptotic change-point model assuming a local
alternative for the probabilities resp. jump heights.

To simplify our exposition, we introduce a maximum sample size $N$
where monitoring stops in any case. Let us also rescale time by the
transformation $ t \mapsto \trunc{Nt} $, t $ \in [0,1] $, where $
\trunc{x} $ denotes the largest integer smaller or equal to $ x $, $
x \in \R $. In the sequel, the current time point $n$ will
correspond to $ t $, i.e., $ n = \trunc{Nt} $.

Define the process
\[
  \calJ_N(t) = \frac{1}{\sqrt{N}} \sum_{i=\trunc{Nt} - M}^{\trunc{Nt}-1}
  (Z_i - p_0), \qquad t \in [(M+1)/N,1].
\]
Note that $ \calJ_N( n/N ) $ is equal to the statistic $ J_n $
centered at its in-control expectation and scaled by $ N^{-1/2} $.
Now, the truncated version of the upper control chart of the last
section, which gives a signal if $ J_n $ exceeds \UCL, corresponds
to the stopping time
\[
  S_N = \min \{ M+1 \le n \le N : J_n > M p_0 + k \sqrt{ M p_0(1-p_0)
  }\}.
\]
We can represent $ S_N $ via the process $ \calJ_N(t) $. Indeed,
we have
\begin{equation}
\label{StopRule}
  S_N = N \inf \left\{ t \in [(M+1)/N, 1] : \calJ_N(t) > k \sqrt{ \frac{M}{N} p_0(1-p_0) }
  \right\}, \qquad N \ge 1.
\end{equation}
For the asymptotic framework in this section, let us assume that the
buffer length, $M $, is chosen as a $\N $-valued function of $ n =
\trunc{Nt}$, i.e., $ M = M_{\trunc{Nt}} $, satisfying the growth
condition
\begin{equation}
\label{CondM}
  \frac{ M_{\trunc{Nt}} }{ N } \to M(t),
\end{equation}
as the maximum sample size $N$ tends to $ \infty $. Here $M :
[0,1] \to [0,1] $ is a non-decreasing function which is continuous
on $ (0,1] $ with $ M(0) = 0 $. We will call $ M ${\em asymptotic
buffer length (strategy)}. Condition (\ref{CondM}) ensures that,
asymptotically, the buffer length $M$ is not too small compared to
$N$.

To ensure that the buffer is not longer than the available time
series, we impose the following condition.

{\bf Assumption (N):} The buffer length strategy $ M :[0,1] \to
[0,1] $ satisfies the {\em natural condition}
\[
  M(t) \le t \qquad \text{for all $ t \in [0,1] $}.
\]

We shall show that under the following assumption the modified chart
is superior to the classic one.

{\bf Assumption (M):} The buffer length strategy satisfies the {\em
modifier condition}, if
\begin{equation}
\label{ModifierCondition}
  M(t) < t \qquad \mbox{for all $t \in (0,1]$},
\end{equation}

Let us now consider some examples.

\begin{example} Put $M(0) = 0 $ and $ M_{\trunc{Nt}} = \trunc{\xi t N} $, $ t \in
(0,1] $, for some $ \xi \in (0,1] $. Obviously, the natural
condition (N) is satisfied, iff. $ \xi < 1 $. Particularly, the
classic $N\, p$ chart is given by $ M_{\trunc{Nt}} = \trunc{Nt} $, $
t \in [0,1] $, thus corresponding to $ \xi = 1 $ and $ M(t) = t $, $
t \in [0,1] $.
\end{example}

The following example considers the case that the buffer lengths $
M_n $ are constant with respect to $ n $.

\begin{example}
\label{Ex2} Suppose $ M_{\trunc{Nt}} = \trunc{ \eta N } $ for some
constant $ \eta \in (0,1] $. For $ t \in [0, \eta/N] $ the available
data $ Y_1, \dots, Y_{\trunc{Nt}} $ do not fill the buffer. One may
assume that pre-run data $ Y_{-M+1}, \dots, Y_0 $ are available.
However, to ensure a fair comparison with the classic $N \, p$
chart, let us consider the choice
\[
  M_{\trunc{Nt}} = \left\{
    \begin{array}{ll}
      0, \qquad & \trunc{Nt}  < \trunc{N \eta}, \\
      \trunc{\eta N}, \qquad & \trunc{Nt} \ge N \eta,
    \end{array} \right.
\]
yielding $ M(t) = \eta \eins_{[\eta,1]}(t) $, $ t \in [0,1] $.
Alternatively, one may set
\[
  M_{\trunc{Nt}} = \min( \trunc{Nt}, \trunc{N \eta} )
\]
yielding $ M(t) = \min(t,\eta) $. Now the modified chart does not
require historical data at the beginning. It starts as the classic
chart and is modified as time proceeds to catch small late changes
better.
\end{example}

Let us now consider an appropriate asymptotic change-point model
for a small jump at location $q$. Assume that
\begin{equation}
\label{CPModel}
  \mu_{Ni} = \E(Z_i) = \left\{ \begin{array}{ll}
    p_0, \qquad i < q = \trunc{N\vartheta}, \\
    p_1, \qquad i \ge q, \end{array} \right.
\end{equation}
for some constant $ \vartheta \in (0,1) $ which specifies the
fraction of the maximum sample size $N$ where the jump occurs.
We model the out-of-control probability $ p_1 $ as a sequence of
local alternatives given by
\[
  p_1 = p_{N1} = p_0 + \Delta/\sqrt{N},
\]
such that $ \Delta = \sqrt{N}(p_1-p_0) > 0 $.

Note that this model yields a triangular array of observations,
\[
  Z_{Ni}, \quad 1 \le i \le N, \quad N \ge 1,
\]
where for each $N$ the random variables $ Z_{N1}, \dots, Z_{NN} $
are independent with $ \E(Z_{Ni}) = p_0 $ for $ 1 \le i < q $ and
$ \E(Z_{Ni}) = p_{N1} $ for $ q \le i \le N $. Below we shall drop
the independence assumption.

\begin{remark}
For our purposes it is appropriate to formulate the change-point
model in terms of the probabilities $p_0 $ and $ p_1 $, but let us
briefly discuss how it relates to a model for the jump height $m$.
Assume the underlying probability density $f(x)$ is continuous and
bounded in a neighborhood of $0$. If we consider a local alternative
model for the jump height where $ m_N = \Delta_m / \sqrt{N} $ for a
positive constant $ \Delta_m $, (\ref{prob1}) and the mean value
theorem give
\[ p_1 - p_0 = f( \xi_N ) \Delta_m / \sqrt{N} \]
for points $ \xi_N $ between $ 0 $ and $ \Delta_m / \sqrt{N} $.
Thus, in this case \[ p_1 = p_0 + (f(0) + o(1))\Delta_m/\sqrt{N}.
\]
\end{remark}

In the sequel, $ B(t) $, $ t \in [0,1] $, denotes a standard
Brownian motion with $ B(0) = 0 $, i.e., a centered Gaussian process
with covariance function $ \Cov( B(s), B(t) ) = \min(s,t) $, $ s, t
\in [0,1] $. The process $ \calJ_N(t) $, $ t \in [0,1] $, is an
element of the Skorohod space $ D[0,1] $ of all functions $ f:[0,1]
\to \R $ which are right-continuous with existing limits from the
left. We denote distributional convergence (weak convergence) for a
sequence $ \{ X, X_n \} \subset D[0,1] $ by $ X_n \Rightarrow X $,
as $ n \to \infty $.
For details we refer to \cite{Bil91} and \cite{Sho00}.


Our main result works under very general assumptions. Indeed, it
just requires that the random variables $ Z_{Ni} - \mu_{Ni} $ form a
martingale difference array with $ \E( Z_{Ni}^r | \calF_{N,i-1} ) =
\mu_{Ni} $ for all $i$ and $ r = 1,2 $, for some filtration $ \{
\calF_{Ni} \} $. In this case, the expectation in (\ref{CPModel}) is
replaced by the conditional expectation $ \E( Z_{Ni} | \calF_{N,i-1}
) $. Recall that an array $ \{ X_{n,m} : 1 \le m \le n_k, n \ge 1 \}
$ of random variables defined on a common probability space $
(\Omega, \calF, \PP) $ is called {\em martingale difference array}
with respect $ \{ \calF_{n,m} \} $, if $ \{ \calF_{n,m} \} $ forms a
{\em filtration}, i.e.,
 \[
  \calF_{n,0} = \{ \emptyset, \Omega \} \subset \calF_{n,1}
  \subset \cdots \subset \calF_{n,n_k} \subset \calF,
\]
each $ X_{n,m} $ is $ \calF_{n,m} $-measureable, and $ \E( X_{n,m} |
\calF_{n,m-1} ) = 0 $ for all $ 1 \le m \le n_k $ and $ n \ge 1 $.

The martingale difference assumption is a natural approach to deal
with time series. However, it is also suited and general enough to
treat (locally) dependent image data, as demonstrated by the
following example working with {\em sliced rectangular
neighborhoods}.

\begin{example} {\sc (A Model for Locally Dependent Image Data)}\\
Suppose each column of an image consisting of $I$ columns and $J$
rows is analyzed from bottom to top. Assume the origin $(0,0)$
corresponds to the lower left corner and the pixels are denoted by
$ (i,j) \in \calI \times \calJ = \{ 0, \dots I \} \times \{ 0,
\dots, J \} $ for integers $ I, J $. Let $ \{ \xi_{ij} : (i,j) \in
\calI \times \calJ \} $ be an array of i.i.d. random variables
with common d.f. $F$ satisfying $ \E(\xi_{ij}) = 0 $ and $ \Var(
\xi_{ij}^2 ) = 1 $ for all $ (i,j) \in \calI \times \calJ $,
representing the background noise of an image. For $ h \ge 1 $
define a {\em sliced $h$-neighborhood} for the pixel $ (i,j) $ by
\[
  \calN_{ij} = \{ (k,l) \in \calI \times \calJ : (k=i \wedge l \le
  j) \vee (1 \le |i-k| \le h \wedge l \le j + h ) \}
\]
and denote by $ \Xi_{ij} = \{ \xi_{kl} : (k,l) \in \calN_{ij} \} $
the corresponding set of $ \xi_{kl} $'s. $  \calN_{ij} $ is a
rectangle with width $ 2h+1 $ and height $ j+h $, sliced along the
line from $ (i,j) $ to $ (i,j+h) $. Then $ \calN_{i1} \subset
\cdots \subset \calN_{iJ} $, and consequently the family
\[
  \calF_{i0} = \{ \emptyset, \Omega \}, \quad
  \calF_{ij} = \sigma( \Xi_{ij} ) = \sigma( \xi_{kl} : (k,l) \in \calN_{ij} ),
\]
defines a filtration. For what follows, notice that $ \xi_{ij} $ is
not an element of the set $ \Xi_{i,j-1} $. Let us now assume that
the errors disturbing the true image are given by the model
equations
\[
  \ve_{ij} = h_{ij} \xi_{ij}, \qquad (i,j) \in \calI \times \calJ,
  \qquad j = 2, \dots, J,
\]
for $ \calF_{i,j-1} $-measureable random variables $ h_{ij} $ with
existing second moments. Then $ h_{ij} = H_{ij}( \Xi_{i,j-1} ) $ for
functions $ H_{ij} $. Obviously, $ \ve_{ij} $ is $ \calF_{ij}
$-measureable and, since $ \xi_{ij} $ is independent from the random
variables of the set $ \Xi_{i,j-1} $, we have $ E( \xi_{ij} |
\calF_{i,j-1} ) = E(\xi_{ij}) = 0 $ yielding
\[
 \E( \ve_{ij} | \calF_{i,j-1} ) = h_{ij} E(\xi_{ij}) = 0.
\]
Thus, $ \{ \ve_{ij} : (i,j) \in \calI \times \calJ \} $ is a
martingale difference array, and $ \{ \ve_{ij} : j \in \calJ \} $
is a martingale difference sequence with respect to $ \{
\calF_{ij} : j = 0, \dots, J \} $ for each $ i \in \calI $. Since
\[
  \Var( \ve_{ij} | \calF_{i,j-1} ) = h_{ij}^2,
\]
$ h_{ij}^2 $ is the conditional variance given the neighboring
pixels. Particularly, $ h_{ij}^2 $ may depend on the noise levels of
these neighboring pixels. Recall that when the $k^{th}$ column is
analyzed, $ Z_{Ni} $ is given by $ Z_{Ni} = \eins( \ve_{ki} \le 0 )
$ for $ i = 1, \dots, N = J $. We have
\[
  \E( Z_{Ni} | \calF_{k,i-1} ) = \PP( h_{ik} \xi_{ik} \le 0 | \calF_{k,i-1}) = F(0/h_{ik}) =
  p_0 = 1/2.
\]
and $ \Var(Z_{Ni} | \calF_{k,i-1} ) = p_0(1-p_0) $. Consequently,
the random variables $ Z_{Ni} - p_0 $, $ i = 1, \dots, N $, also
form a martingale difference array with respect to the filtration $
\{ \calF_{ki} : i = 1, \dots, N \} $ with common conditional
variance $ p_0(1-p_0) $.
\end{example}

We are now in a position to formulate our main result concerning the
weak convergence of the process $ \calJ_N(t) $ and the corresponding
central limit theorem for the modified chart.

\begin{theorem}
\label{main1} Suppose (N) and that the random variables $ \xi_{Ni}^*
= (Z_{Ni}-\mu_{Ni}))/\sqrt{ \mu_{Ni}(1-\mu_{Ni}) } $ form a
martingale difference array with respect to some filtration $
\calF_{Ni} $, such that
\[
  \E( \xi_{Ni}^* | \calF_{N,i-1} ) = 0 \qquad \text{and} \qquad \Var( \xi_{Ni}^* | \calF_{N,i-1} ) =
  1,
\]
for all $ 1 \le i \le N $, $ N \ge 1 $. Then the following
conclusions hold true.
\begin{itemize}
\item[(i)] If there is no change-point, the process $ \calJ_N $
converges weakly,
\[
  \calJ_N(t) \Rightarrow \eta_0 [ B(t) - B( t - M(t) ) ],
\]
as $N \to \infty $, where
\[
  \eta_0^2 = \lim_{N \to \infty} Var\biggl( N^{-1/2} \sum_{i=1}^{N} (Z_{Ni} -
  \E(Z_i)) \biggr) = p_0(1-p_0).
\]
The normed stopping time converges in distribution,
\[
  S_N/N \stackrel{d}{\to} \tau_M
\]
where
\[
  \tau_M = \inf \{ t \in [0,1] : B(t) - B(t-M(t)) > k \sqrt{M(t)} \}
\]
\item[(ii)] Under the local change-point model (\ref{CPModel}), the process $ \calJ_N $
converges weakly,
\[
  \calJ_N(t) \Rightarrow \calJ^{(1)}_M(t) =
  \left\{
  \begin{array}{ll}
    \eta_0 [B(t) - B(t - M(t))],                         & t < \vartheta, \\
    \eta_0 [B(t) - B(t - M(t))] + (t - \vartheta ) \Delta, & \vartheta \le t < \vartheta + M(t), \\
    \eta_0 [B(t) - B(t - M(t))] + M(t) \Delta,           & \vartheta + M(t) \le t,
  \end{array}
  \right.
\]
as $N \to \infty $. The normed stopping time converges in
distribution,
\[
  S_N/N \stackrel{d}{\to} \tau_M^{(1)} = \inf \{ s \in [0,1] : \calJ^{(1)}_M(s) > k \sqrt{M(s)} \eta_0 \}.
\]
\end{itemize}
\end{theorem}

\begin{remark} Notice that the standard i.i.d. setting, where it is
assumed that $ Z_{N1}, \dots, Z_{NN} $ are independent and
identically distributed Bernoulli variables with success
probability $ p_0 $, is covered as a special case.
\end{remark}

The above theorem says that, asymptotically, the control chart
behaves as the stopping time $ \tau_M $ which is driven by the
stochastic process
\[
  \calV(t) = B(t) - B( t - M(t) ).
\]
Notice that the one-dimensional marginals of $ \calV(t) $ are
distributed as $ B(M(t))$. Further, for $ s \le t $ we have
\[
  E \calV(s) \calV(t) = \left\{
    \begin{array}{cc}
      0, \qquad  &s-M(s) \le s \le t-M(t) \le t, \\
      s-t+M(t), \qquad & s-M(s) \le t - M(t) \le s \le t, \\
      M(s), \qquad & t-M(t) \le s -M(s) \le s \le t.
    \end{array} \right.
\]
For small values of $ |s-t| $, i.e., locally, the process $ V(t) $
behaves similar as the process $ B(M(t)) $, if $ M(t) $ is a
smooth function.

The above theoretic results explain the benefits from using the
modified binary chart: Assume (M) and suppose a signal is given at
time $ t \in [\vartheta, \vartheta + M(t)) $ where $ \eta \le
\vartheta $ (cf. Example~\ref{Ex2}.) In this case
\[
  \frac{ \calV(t) }{ \sqrt{M(t)} } + \frac{ (t-\vartheta) \Delta }{
  \eta_0 \sqrt{M(t)} } > k.
\]
Right before the threshold $k$ is hit, the behavior of the random
part of the left hand side can be approximated by the process $
B(M(t))/\sqrt{M(t)} $, which has expectation $0$, variance $1$ for
any function $ M(t) $, and covariance function
\[
 (s,t) \mapsto \frac{ \min( M(s), M(t) ) }{ \sqrt{M(s) M(t)} }.
\]
For small values of $ |s-t| $ and smooth $M(t)$ this is
approximately a Brownian motion. Consider the drift term $
(t-\vartheta) \Delta / (\eta_0 \sqrt{M(t)} ) $, which mainly yields
the detection power. The modifier condition (M) ensures that the
drift term is strictly larger than the drift term for the case $
M(t) = t $ corresponding to the classic $N\, p$ chart. This explains
the superior performance of the modified chart for small jumps.

If the change was not detected until time $ \vartheta + M(t) $, a
signal is given if
\[
  \frac{ \calV(t) }{ \sqrt{M(t)} } + \frac{ \sqrt{M(t)} \Delta }{\eta_0}  > k.
\]
For the random part the same arguments as given above apply. But now
under condition (M) the drift term is strictly smaller than the
drift term for the case $ M(t) = t $. We may summarize that the
limit theorem indicates that the modified $p$ chart is preferable to
detect very small jumps right after the change-point.

Also notice that Theorem~\ref{main1} yields well defined limit
distributions for small jumps of the order $ N^{-1/2} $. Clearly,
for jumps of higher order, the drift diverges and dominates the
random part, such that the beneficial effect of the function $ M(t)
$ is not visible.

\section{Practical issues of control chart design}
\label{PracticalIssues}

Unlike the classic $N\,p$-chart, the modified chart has two tunable
parameters, namely, $M$ and $k$, which should be carefully selected
in order to ensure small out-of-control ARLs (average run length to
detection) under the constraint that the in-control ARL (average run
length to false alarm) is not smaller than a given level.

We will now summarize our experience on tuning this chart by
simulations, which are justified to some extent by the theoretical
results presented in the previous section. The major issue is how to
select the control limit.
\begin{itemize}
    \item[(i)] In practice, the $ 3 \sigma$ rule is often advocated, i.e.,
    $k=3$. However, this is not advisable here, since it
    leads to excessively long in-control ARLs. For our control
    chart, the in-control ARL also depends on the buffer length
    $M$. Selecting $k=2.34$ and $M=9$ we get first reasonable
    in-control ARL about $500$.
    \item[(ii)] For a given buffer length
    the same in-control ARL is attained for $k$ from
    a certain relatively long interval. This is due to the fact
    that $J_n$ is always an integer.
    \item[(iii)]
    Analysis of Figure~\ref{figARLInC}, where $\log$ of
    in-control ARL is plotted as a function of $k$
    for different buffer lengths, reveals that
    it is advisable to select $k$ at the left end
    of that interval. That choice ensures the specified
    in-control ARL and  minimizes the distance $ \UCL - \LCL $.
\end{itemize}

In view of these remarks we suggest the following practical approach
to select the parameters $M$ and $k$ of the chart.

\begin{enumerate}
        \item Select a desired in-control ARL, e.g.,
        equal to 370.
        \item Select the buffer length $M>1$. A discussion on
        selecting $M$ is presented below.
        \item For a practical application one may simulate the
        in-control ARL for $k$ varying from 1 to 3.
        It is not difficult to find a reasonable $k$ in this way,
        but determining exactly the smallest
        $k$, which guarantees the specified in-control ARL is a computationally
        demanding task.
    \end{enumerate}

For the reader's convenience Table~\ref{tabARLInC} summarizes some
pairs $(M,k) $ with minimal $ k $ (accuracy $0.01$) ensuring an
in-control ARL of approximately $435$. Notice that in general the
fact that $ J_n $ is integer-valued prevents the construction of a
control chart with in-control ARL being equal to the target
in-control ARL.

\begin{figure}
\begin{center}
\fbox{\includegraphics[width=0.925\textwidth]{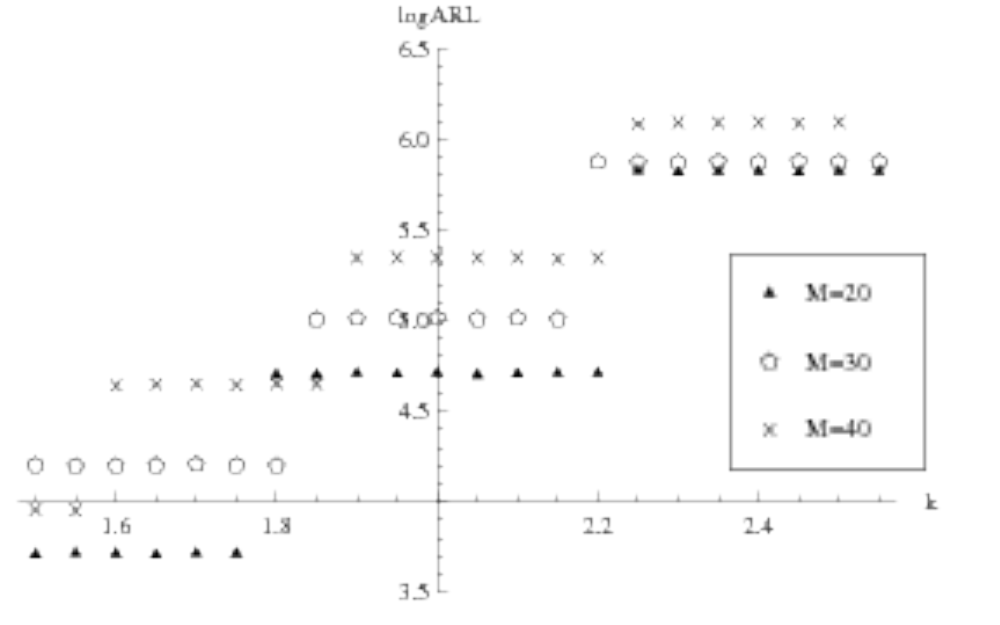}}
\caption{Dependence of the logarithm of the in-control ARL on the
threshold $k$ for different buffer sizes $M$. The results were
obtained for Gaussian $N(0,1)$ errors by averaging
 $10^4$ simulation runs.}\label{figARLInC}
\end{center}
\end{figure}

One may also select $M$ to minimize the out-of-control ARL for a
given jump height $m$. Figure~\ref{figARLOutCvsM} indicates that for
jump heights $m=0.25$, $m=0.5$, and $m=0.75$ there exist optimal
buffer lengths $M$. The choices $M=71$, $M=28$, $M=23$ are optimal
for $m=0.25$, $m=0.5$, and $m=0.75$, respectively, taking into
account that the selection was made among a rather limited number of
buffer lengths. Clearly, an exhaustive search may yield slightly
better results. Note, however, that for $m=1$ the plot is increasing
and one might expect that the best choice is for $M<12$, but in this
region one can not attain in-control ARL of order $435$.

\begin{figure}
\begin{center}
\fbox{\includegraphics[width=0.925\textwidth]{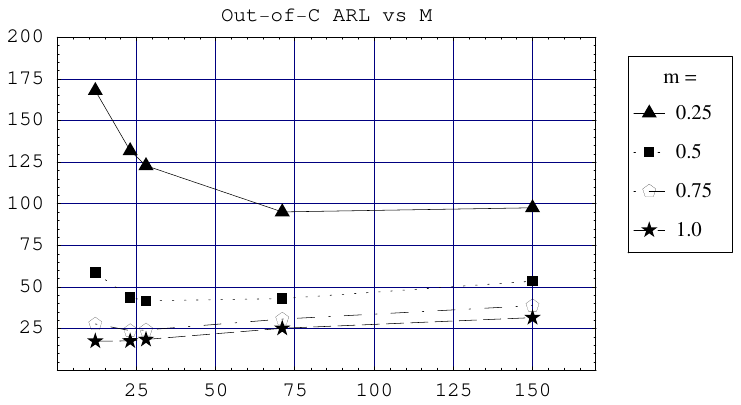}}
\caption{Dependence of the out-of-control ARL as a function of the
buffer size $M$ for different jump heights $m$ and normal
errors.}\label{figARLOutCvsM}
\end{center}
\end{figure}
\renewcommand{\arraystretch}{1.75}
\begin{table}
  \centering
\begin{tabular}{|r|c|c|c|c|c|c|c|c|}  \hline
  $M\, =$      & 12   & 23   & 28   & 71  & 90  & 150 & 212 & 441 \\ \hline
  $k\, =$      & 2.31 & 2.30 & 2.27 & 2.02& 2.0 & 1.8 & 1.65& 1.39 \\  \hline
  ARL${}_0\, =$& 395  & 415  & 423  & 411 & 450 & 452 & 440 & 456  \\   \hline
\end{tabular}
  \caption{Pairs of parameters $(M,k)$ of the proposed chart
  ensuring an in-control ARL or the order $435$.}\label{tabARLInC}
\end{table}

\section{Simulation studies}
\label{Sim}

We performed extensive simulations aiming at the following issues.
Firstly, we were interested in identifying pairs of the buffer
length $M$ and the threshold $k$ ensuring a specified in-control ARL
(at least approximately). Secondly, we investigated the
out-of-control ARL for various jump heights, when the underlying
observations are normally distributed. Third, we compared the binary
chart with other charts for the case of normally distributed error
terms, focusing on the out-of-control ARL as a performance measure.
Finally, we studied the behavior of the out-of-control ARL for the
binary chart when the errors are non-normal.

The simulation results are given in the tables below. All the
results were obtained by averaging 30000 simulation runs. Simulated
jump occured at time zero and the buffer was fed up by in-control
pre-run observations. The results of simulation studies can be
summarized as follows.
\begin{itemize}
\item[(i)] For Gaussian errors and an out-of-control ARL fixed at
435, our chart with buffer length $M=150$ (see Table~\ref{G435Long})
provides shorter out-of-control ARL's than CUSUM, Optimal EWMA,
Shewhart-EWMA, GEWMA and GLR (see~\cite{AnnStat2004} for
definitions), provided the jump is small. To be precise, the
out-of-control ARL of our chart is about $243$ for a jump $m = 0.1\,
\sigma$, and about $97$ for $m = 0.25\, \sigma$, while for the above
mentioned charts we have ARL's between $295$ and $324$ and between
$105$ and $110$, respectively. Simultaneously, the dispersion of the
RL time of our chart is considerably smaller and equals $172$ for $m
= 0.1\, \sigma$ and about $59$ for $m=0.25\, \sigma$, while for the
charts discussed in \cite{AnnStat2004} we have RL time dispersions
of the orders $267$-$324$ and $79$-$102$, respectively. \item[(ii)]
Qualitatively the same pattern can be observed when the
out-of-control ARL is fixed at $840$ and errors are Gaussian (see
Table~\ref{G840Long} and \cite{AnnStat2004}). \item[(iii)] When the
jump is larger than $0.5\, \sigma$, the proposed chart is much
slower than the above mentioned charts, but this shortcoming can
easily be handled by applying several charts simultaneously and
claiming an alarm when one of them gives a signal. \item[(iv)] The
proposed chart retains its advantages in the range of small jumps
when the errors are  double exponentially distributed and even
behaves quite well for difficult distributions as the Cauchy one
(see Table~\ref{nonG435}).
\end{itemize}

\begin{table}[h]
  \centering
\mbox{ \noindent\begin{tabular}{|c|c|c|}
  \hline
\multicolumn{3}{|c|}{ $M$= 12, $k=$2.31 } \\ \hline Jump &   ARL &
RL Disp, \\ \hline
0    &   395.27 & 171.09  \\ \hline 0.1  &   328.33 & 144.18  \\
\hline 0.25 &   168.09 &  72.47   \\ \hline 0.5  &    58.65 &  24.52
\\ \hline 0.75 &    27.84 &  10.91   \\ \hline 1    &    17.51 &
6.35   \\ \hline 1.25 &    12.98 &   4.41   \\ \hline 1.5  & 10.96 &
3.54   \\ \hline 1.75 &    10.00 &   3.14   \\ \hline 2 &     9.46 &
2.94   \\ \hline 2.25 &     9.19 &   2.84   \\ \hline
2.5  &     9.09 &   2.80   \\ \hline 2.75 &     9.05 &   2.79   \\
\hline 3    &     9.01 &   2.77   \\ \hline
\end{tabular}
\begin{tabular}{|c|c|c|}
  \hline
\multicolumn{3}{|c|}{ $M$= 23, $k=$2.3 } \\ \hline Jump &   ARL &
RL Disp, \\
\hline 0   &  415.66 & 181.42 \\ \hline 0.1 &  305.80 & 133.14  \\
\hline 0.25&  131.89  & 56.00  \\ \hline 0.5 &   43.78 & 17.08
\\ \hline 0.75&   23.76 &   8.35  \\ \hline 1   & 17.60 & 5.76
\\ \hline 1.25&   14.99 &   4.76  \\ \hline 1.5 &   13.66 &   4.31
\\ \hline 1.75&   12.88 &   4.05  \\ \hline
2   &   12.45 &   3.91      \\ \hline 2.25&   12.26 &   3.84  \\
\hline 2.5 &   12.12 &   3.80      \\ \hline 2.75&   12.04 &   3.77
\\ \hline 3   &   11.96 &   3.75      \\ \hline
\end{tabular}
\begin{tabular}{|c|c|c|}
  \hline
\multicolumn{3}{|c|}{ $M$= 28, $k=$2.27 } \\ \hline Jump &   ARL &
RL Disp, \\ \hline
0   & 423.12  &  185.75   \\ \hline 0.1 & 303.43  &  133.17    \\
\hline 0.25& 122.90  &   51.72    \\ \hline 0.5 &  41.66  &   15.73
\\ \hline 0.75&  24.18  &    8.27     \\ \hline 1   &  18.53  &
6.00    \\ \hline 1.25&  16.10  &    5.12     \\ \hline 1.5 &  14.76
&    4.68    \\ \hline 1.75&  13.99  &    4.42     \\ \hline 2   &
13.54  &    4.28    \\ \hline 2.25&  13.25  &    4.18     \\ \hline
2.5 &  13.14  &    4.14    \\ \hline 2.75&  12.96  &    4.09     \\
\hline 3   &  13.09  &    4.12    \\ \hline
\end{tabular}}\\[0,5cm]
\caption{Binary chart applied to observations with
 Gaussian errors. Chart tuned to in-control ARL about $435$.
 Short buffer length.}\label{G435Short}
\end{table}
\begin{table}[h]
\centering \mbox{
\begin{tabular}{|c|c|c|}
  \hline
\multicolumn{3}{|c|}{ $M$= 71, $k=$2.02 } \\ \hline Jump &   ARL &
RL Disp, \\ \hline 0   & 411.23& 301.39    \\ \hline 0.1 & 254.91&
182.71    \\ \hline 0.25&  95.12&  60.68    \\ \hline 0.5 &  43.03&
23.33    \\ \hline 0.75&  30.75&  16.08    \\ \hline 1   &  25.23&
13.06    \\ \hline 1.25&  22.11&  11.39    \\ \hline 1.5 &  20.28&
10.41    \\ \hline 1.75&  19.22&   9.83    \\ \hline 2   &  18.61&
9.50    \\ \hline 2.25&  18.13&   9.25    \\ \hline 2.5 &  17.91&
9.14    \\ \hline 2.75&  17.83&   9.08    \\ \hline 3   &  17.69&
9.03    \\ \hline
\end{tabular}
\begin{tabular}{|c|c|c|}
  \hline
\multicolumn{3}{|c|}{ $M$= 150, $k=$1.8 } \\ \hline Jump &   ARL &
RL Disp, \\ \hline 0   & 452.05& 337.19  \\ \hline 0.1 & 243.54&
172.58   \\ \hline 0.25&  97.58&  58.68   \\ \hline 0.5 &  53.50&
29.52   \\ \hline 0.75&  38.80&  21.12   \\ \hline 1   &  31.60&
17.07   \\ \hline 1.25&  27.71&  14.87   \\ \hline 1.5 &  25.20&
13.50   \\ \hline 1.75&  23.82&  12.74   \\ \hline 2   &  23.10&
12.30   \\ \hline 2.25&  22.64&  12.05   \\ \hline 2.5 &  22.31&
11.86   \\ \hline 2.75&  22.17&  11.80   \\ \hline 3   &  22.15&
11.77   \\ \hline
\end{tabular}
\begin{tabular}{|c|c|c|}
  \hline
\multicolumn{3}{|c|}{ $M$= 212, $k=$1.65 } \\ \hline Jump &   ARL &
RL Disp, \\ \hline 0   &  440.32& 334.70  \\ \hline 0.1 &  234.27&
166.92   \\ \hline 0.25&  101.26&  60.62  \\ \hline 0.5 &   56.87&
32.41   \\ \hline 0.75&   41.30&  23.18  \\ \hline 1   &   33.77&
18.76   \\ \hline 1.25&   29.38&  16.24  \\ \hline 1.5 &   26.92&
14.81   \\ \hline 1.75&   25.38&  13.93  \\ \hline 2   &   24.57&
13.48   \\ \hline 2.25&   24.17&  13.19  \\ \hline 2.5 &   23.76&
13.00   \\ \hline 2.75&   23.70&  12.95  \\ \hline 3   &   23.70&
12.94   \\ \hline
\end{tabular}}\\[0,5cm]
\caption{Binary chart applied to observations with
 Gaussian errors. Chart tuned to in-control ARL about $435$.
 Moderate and long buffer length.}\label{G435Long}
\end{table}

\begin{table}[h]
\centering \mbox{
\begin{tabular}{|c|c|c|}
  \hline
\multicolumn{3}{|c|}{ $M$= 111, $k=$2.19 } \\ \hline Jump &   ARL &
RL Disp, \\ \hline 0   & 836.64& 370.04  \\ \hline 0.1 & 398.04&
170.64   \\ \hline 0.25& 122.34&  46.80    \\ \hline 0.5 & 57.56 &
19.21     \\ \hline 0.75& 41.71 &  13.77    \\ \hline 1   & 34.13 &
11.17     \\ \hline 1.25& 29.78 &   9.72    \\ \hline 1.5 & 27.12 &
8.84     \\ \hline 1.75& 25.66 &   8.35    \\ \hline 2   & 24.76 &
8.04     \\ \hline 2.25& 24.24 &   7.88    \\ \hline 2.5 & 24.00 &
7.78     \\ \hline 2.75& 23.97 &   7.78    \\ \hline 3   & 23.79 &
7.73     \\ \hline
\end{tabular}
\begin{tabular}{|c|c|c|}
  \hline
\multicolumn{3}{|c|}{ $M$= 131, $k=$1,84 } \\ \hline Jump &   ARL &
RL Disp, \\ \hline
0   &  841.83 & 370.73   \\ \hline 0.1 &  399.22 & 171.86    \\
\hline 0.25&  124.06 &  46.98    \\ \hline 0.5 &   57.63 &  19.18
\\ \hline 0.75&   42.10 &  13.83    \\ \hline 1   &   34.18 &  11.20
\\ \hline 1.25&   29.60 &   9.66    \\ \hline 1.5 &   27.18 &   8.84
\\ \hline 1.75&   25.61 &   8.32    \\ \hline 2   &   24.66 &   8.02
\\ \hline 2.25&   24.21 &   7.87    \\ \hline 2.5 &   23.96 &   7.77
\\ \hline 2.75&   23.78 &   7.74    \\ \hline 3   &   23.69 &   7.70
\\ \hline
\end{tabular}
\begin{tabular}{|c|c|c|}
  \hline
\multicolumn{3}{|c|}{ $M$= 453, $k=$1.35 } \\ \hline Jump &   ARL &
RL Disp, \\ \hline 0   & 840.02& 650.13\\ \hline 0.1 & 337.79&
226.64\\ \hline 0.25& 149.54 & 87.46\\ \hline 0.5  & 82.94 & 47.32\\
\hline 0.75 & 59.05 &
33.34\\ \hline 1    & 48.22 & 27.03\\ \hline 1.25 & 42.12 & 23.43\\
\hline 1.5  & 38.14 & 21.23\\ \hline 1.75 & 36.05 & 20.03\\ \hline 2
& 34.93 & 19.38\\ \hline 2.25 & 34.33 & 18.98\\ \hline 2.5  & 33.79
& 18.71\\ \hline 2.75 & 33.36 & 18.48\\ \hline 3    & 33.55 &
18.57\\ \hline
\end{tabular}}\\[0,5cm]
\caption{Binary chart applied to observations with
 Gaussian errors. Chart tuned to in-control ARL about $840$.
 Moderate and long buffer length.}\label{G840Long}
\end{table}

\begin{table}
\centering
\begin{tabular}{|c|c|c|}
  \hline
\multicolumn{3}{|c|}{ Laplace (DblExp)} \\ \hline
\multicolumn{3}{|c|}{ $M$= 40, $k=$2.22 } \\ \hline Jump &   ARL &
RL Disp, \\ \hline 0   &  437.69& 315.91  \\ \hline 0.1 &  191.35&
133.31  \\ \hline 0.25&   59.51&  37.02  \\ \hline 0.5 &   28.51&
15.02  \\ \hline 0.75&   22.04&  11.13  \\ \hline 1   &   19.33&
9.69  \\ \hline 1.25&   17.70&   8.84  \\ \hline 1.5 &   16.85& 8.37
\\ \hline 1.75&   16.15&   8.02  \\ \hline 2   &   15.78& 7.83  \\
\hline 2.25&   15.55&   7.69  \\ \hline 2.5 &   15.34& 7.59  \\
\hline 2.75&   15.19&   7.52  \\ \hline 3   &   15.07& 7.46  \\
\hline
\end{tabular}
\begin{tabular}{|c|c|c|}
  \hline
\multicolumn{3}{|c|}{Cauchy} \\ \hline \multicolumn{3}{|c|}{ $M$=
28, $k=$2.28 } \\ \hline Jump &   ARL    & RL Disp, \\ \hline
0    & 420.79 & 300.12   \\ \hline 0.1  & 334.82 & 240.04   \\
\hline 0.25 & 167.28 & 116.28   \\ \hline 0.5  &  64.17 &  41.76
\\ \hline 0.75 &  37.52 &  22.47   \\ \hline 1    &  27.27 &  15.21
\\ \hline 1.25 &  22.70 &  12.03   \\ \hline 1.5  &  20.51 &  10.58
\\ \hline 1.75 &  18.86 &   9.55   \\ \hline 2    &  17.93 &   9.00
\\ \hline 2.25 &  17.23 &   8.58   \\ \hline 2.5  &  16.67 &   8.28
\\ \hline 2.75 &  16.29 &   8.07   \\ \hline 3    &  15.98 &   7.90
\\ \hline
\end{tabular}\\[0,5cm]
\caption{Comparison of ARLs of the binary chart with in-control ARL
~ 435 when applied to non-Gaussian distributions.
 $30,000$ independent
simulation trials.}\label{nonG435}
\end{table}

\section*{Acknowledgments}

The authors thank anonymous referees for constructive remarks
which improved the presentation. Part of the paper was prepared
during a visit of A.~Steland at the Technical University of Wroc\l
aw. The work of E. Rafaj{\l}owicz was supported by a research
grant ranging from 2007 to 2009 from the Ministry of Science and
Higher Education of  Poland.

\appendix

\section{Proof of the main result}

Under the change-point model of Section~\ref{CPM} we are given an
array $ \{ Z_{Ni} : 1 \le i \le N, N \ge 1 \} $ of Bernoulli
variables with conditional expectations $ \E( Z_{Ni} | \calF_{N,i-1}
) = p_0 $ if $ 1 \le i < \trunc{N\vartheta} $, and $ \E( Z_{Ni} |
\calF_{N,i-1} ) = p_{N1} = p_0 + \Delta/\sqrt{N} $ if $ \trunc{N
\vartheta} \le i \le N $, $ N \ge 1 $.

\begin{theorem} (Durrett 2005, Theorem~7.3). \label{ThA} Suppose $ \{ X_{n,m} \} $ is a martingale difference array
with respect to $ \{ \calF_{n,m} \} $. Define
\[
  S_{n,k} = \sum_{i=1}^k X_{n,i}, \qquad V_{n,k} =
  \sum_{1 \le i \le k} \E( X^2_{n,i} | \calF_{n,i-1} ), \qquad 0
  \le k \le n.
\]
If
\begin{itemize}
  \item[(i)] $ V_{n,\trunc{nt}} \to t $ in probability for all $t
\in [0,1] $ and
  \item[(ii)] for all $ \varepsilon > 0 $, $ \sum_{m
\le n} \E( X^2_{n,m} \eins_{\{ |X_{n,m}| > \varepsilon \}} |
\calF_{n,m-1} ) \to 0 $ in probability, \end{itemize} then $
S_{n,\trunc{nt}} \Rightarrow B(t) $, where $B$ denotes a standard
Brownian motion.
\end{theorem}

\begin{proof} (of Theorem~3.2) We first consider the case when there is no change.
Let us introduce the partial sum process,
\[
  Z_N(t) = \sum_{i=1}^{\trunc{Nt}} \xi_{Ni}, \qquad t \in [0,1],
\]
where $ \xi_{Ni} = (Z_{Ni} - p_0)/\sqrt{N p_0(1-p_0)}
 $, $ 1 \le i \le N $. Let us first verify that the array $ \{ \xi_{Ni} : 1 \le i \le N,
N \ge 1 \} $ satisfies the assumptions of Theorem~\ref{ThA}.
Clearly, $ E( \xi_{Ni} | \calF_{N,i-1} ) = 0 $ and
\[
  E( \xi_{Ni}^2 | \calF_{N,i-1} ) = \Var( \xi_{Ni} | \calF_{N,i-1}
  ) = N^{-1},
\]
for all $ 1 \le i \le N $, yielding
\[
  V_{N,\trunc{Nt}} = \sum_{i=1}^{\trunc{Nt}} E( \xi_{Ni}^2 |
  \calF_{N,i-1} ) = \frac{\trunc{Nt}}{N} \to t,
\]
as $ N \to \infty $. The conditional Lindeberg condition is shown as
follows. Since $ E( (Z_{Ni}-p_0)^2 | \calF_{N,i-1} ) \le 1 $, $ 1
\le i \le N $, we obtain for any $ \varepsilon > 0 $
\begin{align*}
  L_N(\varepsilon) &= \sum_{i=1}^N E( \xi_{Ni}^2 \eins( | \xi_{Ni} | > \varepsilon ) | \calF_{N,i-1} )  \\
    & \frac{1}{N} \sum_{i=1}^N E \left( \frac{ (Z_{Ni}-p_0)^2 }{
    p_0(1-p_0) } \eins \left( \frac{ |Z_{Ni} - p_0| }{\sqrt{p_0(1-p_0)}} >
    \varepsilon \sqrt{N} \right) \bigg| \calF_{N,i-1} \right) \\
    & \le \frac{1}{Np_0(1-p_0)} \sum_{i=1}^N P\left( \frac{|Z_{Ni} -
    p_0| }{ \sqrt{p_0(1-p_0)} } > \varepsilon \sqrt{N} \bigg|
    \calF_{N,i-1} \right).
\end{align*}
The conditional Markov inequality yields for $ 1 \le i \le N $
\[
 P\left( \frac{|Z_{Ni} -
    p_0| }{ \sqrt{p_0(1-p_0)} } > \varepsilon \sqrt{N} \bigg|
    \calF_{N,i-1} \right)
 \le
  \frac{1}{\varepsilon^2 N},
\]
which implies
\[
  \lim_{N \to \infty} L_N( \varepsilon ) = 0.
\]
Hence, by Theorem~\ref{ThA}
\[
  Z_N \Rightarrow B, \qquad  N \to \infty.
\]
Now, as will be shown below for a more involved setting,
\begin{align*}
  J_N(t) & = \sqrt{p_0(1-p_0)}\left[ Z_N( t - \frac{1}{N} ) - Z_N(
  t - \frac{M_{\trunc{Nt}}}{N} - \frac{1}{N} ) \right] \\
   & \Rightarrow \eta_0[ B(t) - B(t-M(t)) ],
\end{align*}
as $ N \to \infty $. Having in mind the rule (\ref{StopRule}), we
conclude
\[
  \calJ_N(t) - k \sqrt{ M_{\trunc{Nt} } N^{-1} p_0(1-p_0) } \Rightarrow
  \eta_0[ B( t ) - B( t - M(t) ) ] - k \sqrt{ M(t) } \eta_0, \qquad N \to \infty,
\]
which yields
\[
  S_N/N \stackrel{d}{\to} \inf \{ s \in (0,1] : B(t) - B(t-M(t))  > k
  \sqrt{M(s)} \},
\]
as $ N \to \infty $.

To establish (ii), we consider three cases.

Case 1: $ \trunc{Nt} \le \trunc{N \vartheta} $ is handled as
above.

Case 2: $ \trunc{N \vartheta} < \trunc{Nt} < \trunc{N\vartheta} +
M_{\trunc{Nt}} $. Denote the set of corresponding values of $t$ by
$ \calT_2 $. $ \calJ_N(t) $ equals
\begin{equation}
\label{Decomp}
  \frac{1}{\sqrt{N}} \sum_{i=\trunc{Nt}-M_{\trunc{Nt}}}^{ \trunc{N
  \vartheta}-1} (Z_i-p_0) + \frac{1}{\sqrt{N}}
  \sum_{i=\trunc{N\vartheta}}^{\trunc{Nt}-1} (Z_{Ni}-p_{N1}) +
  \frac{1}{\sqrt{N}} \sum_{i=\trunc{N\vartheta}}^{\trunc{Nt}-1}
  (p_{N1} - p_0).
\end{equation}
Since $ p_1-p_0 = \Delta/\sqrt{N} $, the third term converges
(pointwise) to the continuous function $ \Delta (t-\vartheta) $,
which implies that the convergence is also uniform in $ t \in
[\vartheta, \vartheta+M(t)] $. To handle the random terms put
  \[
    \widetilde{\xi}_{Ni} = \left\{
    \begin{array}{cc}
      ( Z_i - p_0 ) / \sqrt{ p_0(1-p_0)N }, \qquad 0 \le i \le
      \trunc{N \vartheta} - 1, \\
      ( Z_{Ni} - p_{N1} ) / \sqrt{ p_{N1}( 1 - p_{N1} ) N}, \qquad
      \trunc{N \vartheta} \le i \le N.
    \end{array} \right.
    \]
    Again, the conditions of the functional martingale central limit theorem are
    satisfied, such that $ \widetilde{Z}_N(t) = \sum_{i=1}^{\trunc{Nt}}
    \widetilde{\xi}_{Ni} \Rightarrow B(t) $. The first and second
    term in (\ref{Decomp}) are now given by
    \begin{align*}
     & \sqrt{ p_0(1-p_0) }
      \biggl[ \widetilde{Z}_N( \vartheta - \frac{1}{N} ) - \widetilde{Z}_N( t - \frac{M_{\trunc{Nt}}}{N} - \frac{1}{N} )
      \biggr] \\
     & \qquad \quad
        + \sqrt{ p_{N1} (1-p_{N1}) } \biggl[ \widetilde{Z}_N( t - \frac{1}{N} ) - \widetilde{Z}_N( \vartheta - \frac{1}{N} ) \biggr],
    \end{align*}
    which equals $ \varphi_N( \widetilde{Z}_N )(t) $, if we define the sequence of functionals $ \varphi_N : (D[0,1],d) \to (D[0,1],d) $, $ N \ge 1 $,
    by
    \begin{align*}
      \varphi_N(z)(t) &=
      \sqrt{ p_0(1-p_0) } \biggl[z(\vartheta-1/N) - z(t - \frac{ M_{\trunc{Nt}} }{N} - \frac{1}{N} )
      \biggr] \\
      & \qquad
        + \sqrt{ p_{N1}(1-p_{N1}) } \biggl[ z(t - \frac{1}{N}) - z( \vartheta - \frac{1}{N} )
        \biggr].
    \end{align*}
    Also define
    \[
      \varphi(z) = \sqrt{p_0(1 - p_0)}[ z(t) - z( t - M(t) ) ],
      \qquad z \in C[0,1].
    \]
    By linearity, $ \varphi_N $ is uniformly Lipschitz continuous,
    i.e.,
    \[
      \sup_{N \ge 1} \| \varphi_N(z_1) - \varphi_N(z_2) \|_\infty \le L \| z_1 - z_2
      \|_\infty,
    \]
    for all $ z_1, z_2 \in D[0,1] $, where $ L = 2 \sup_{N \ge 1} \sqrt{ p_{N1}(1-p_{N1})
    } < \infty $. Further, since any $ z \in
    C[0,1] $ is uniformly continuous,
    \[
      \| \varphi_N(z) - \varphi(z) \|_\infty \to 0, \qquad N \to
      \infty.
    \]
    Let $ \{ z, z_N \} \subset D[0,1] $ be a sequence with
    $ z_N \to z \in C[0,1]
    $ in the Skorohod metric, which implies $ \| z_N - z \|_\infty
    \to 0 $. Apply the triangle inequality to obtain
    \[
      \| \varphi_N( z_N ) - \varphi(z) \|_\infty \le \| \varphi_N(z_N) -
      \varphi_N(z) \|_\infty + \| \varphi_N( z ) - \varphi( z )
      \|_\infty.
    \]
    The first term is bounded by $ L \| z_N - z \|_\infty \to 0$, $N \to \infty$, and the second one
   tends to $0$ by the uniform Lipschitz continuity. For $ z \in
    C[0,1] $ we have $ \varphi(z)(t) = \sqrt{p_0(1-p_0)} [z(t)
    -z(t-M(t)) ] $. Due to the Shorohod/Dudley/Wichura
    representation theorem, $ \widetilde{Z}_N \Rightarrow B $, $ N \to \infty $, implies that
    there exists a probability space and equivalent version of $ \widetilde{Z}_N $ and $ B $ defined on that new space, which
    we again denote by $ \widetilde{Z}_N $ and $B$, such that $ \| \widetilde{Z}_N - B
    \|_\infty \to 0 $, $N \to \infty $, a.s. The above arguments ensure that
    \[
    \varphi_N(\widetilde{Z}_N)(t) \Rightarrow \varphi(B)(t) = \eta_0 [ B(t) - B( t - M(t) ) ],
    \]
    as $ N \to \infty $.

Case 3: $ \trunc{N\vartheta} + M_{\trunc{Nt}} \le t $ is obvious.

Putting things together yields the result for $ \calJ_N(t) $.
Since the process $ \calJ_M^{(1)} $ is a.s. continuous, we may
further conclude that
\[
  S_N/N \stackrel{d}{\to} \tau_M^{(1)} = \inf \{ t \in [0,1] : \calJ^{(1)}_M(t) > k \sqrt{M(t)} \eta_0 \},
\]
as $ N \to \infty $.
\end{proof}

\end{document}